\documentclass[amsmath,amssymb,superscriptaddress]{revtex4}
\usepackage[english]{babel}
\usepackage{amssymb,latexsym}
\usepackage{amssymb}
\usepackage{amsmath}
\usepackage[latin1]{inputenc}
\usepackage{pst-3d}
\usepackage{latexsym}
\usepackage{pb-diagram}
\usepackage{graphicx}

\definecolor{gris}{cmyk}{0,0,0,0.15}
\usepackage{latexsym,amssymb}
\usepackage[english]{babel}
\usepackage{graphics}
\usepackage{amsmath}
\usepackage{epsfig}
\usepackage{enumerate}
\usepackage{amsthm}
\usepackage{hyperref}
\usepackage{pb-diagram}

\newtheorem{thm}{Theorem}[section]

\newtheorem{lemma}{Lemma}[section]

\newtheorem{definition}{Definition}[section]

\newcommand{\R}{\mathbb{R}}
\usepackage{fancyhdr}
 
\begin{document}

\begin{flushright}CUQM-160~~~~~~\end{flushright}
\title{General comparison theorems for the Klein--Gordon equation\newline in $d$ dimensions}
\author{Richard L. Hall}
\email{richard.hall@concordia.ca}
\affiliation{Department of Mathematics and Statistics, Concordia University,
1455 de Maisonneuve Boulevard West, Montr\'eal,
Qu\'ebec, Canada H3G 1M8}
\author{Hassan Harb}
\email{hassan.harb@concordia.ca}
\affiliation{Department of Mathematics and Statistics, Concordia University,
1455 de Maisonneuve Boulevard West, Montr\'eal,
Qu\'ebec, Canada H3G 1M8}

\begin{abstract}
We study bound-state solutions of the Klein--Gordon equation $\varphi^{\prime\prime}(x) =\big[m^2-\big(E-v\,f(x)\big)^2\big]  \varphi(x),$ 
for bounded vector potentials which in one spatial dimension have the form $V(x) = v\,f(x),$ where $f(x)\le 0$ is the shape of a finite symmetric central potential  that is monotone non-decreasing on $[0, \infty)$ and vanishes as $x\rightarrow\infty.$ Two principal results are  reported.  First, it is shown that the eigenvalue problem in the coupling parameter $v$ leads to spectral functions of the form  $v= G(E)$ which are concave, and at most uni-modal with a maximum near the lower limit $E = -m$ of the eigenenergy $E \in (-m, \, m)$.  This formulation of the spectral problem immediately extends to central potentials in $d > 1$ spatial dimensions.  Secondly, for each of the dimension cases, $d=1$ and $d  \ge 2$, a comparison theorem is proven, to the effect that if two potential shapes are ordered $f_1(r) \leq f_2(r),$ then so are the  corresponding pairs of spectral functions $G_1(E) \leq G_2(E)$ for each of the existing eigenvalues.  These results remove the restriction to positive eigenvalues necessitated by earlier comparison theorems for the Klein--Gordon equation.
\end{abstract}

\maketitle
\noindent{\bf Keywords:~}Klein--Gordon equation, discrete spectrum, comparison theorem  \\  
\noindent{\bf PACS:} 03.65.Pm, 03.65.Ge, 36.20.Kd.
\vskip 0.2 in
\section{Introduction}

The elementary comparison theorem of non-relativistic quantum mechanics states that if two potentials are ordered, then the respective bound-state eigenvalues are correspondingly ordered:
\begin{eqnarray*}
 V_1 \leq V_2\Longrightarrow E_1\leq E_2.
\end{eqnarray*}
 In the non-relativistic case (Schr\"odinger's Equation), this is a direct consequence of the min-max principle since the Hamiltonian $H = -\Delta + V$ is bounded from below, and the discrete spectrum can be characterized variationally \cite{reedsimon}. However, the min-max principle is not valid in a simple form in the relativistic case because the energy operators are not bounded from below \cite{Fr, Gold, Gr}. Regarding the Klein--Gordon equation, since only a few analytical solutions are known, the existence of lower and upper bounds for the eigenvalues is important, and establishing comparison theorems for the eigenvalues of this equation is of considerable interest.  We suppose that the vector potential $V$ is written in the form $V(x)=v\,f(x)$, where  $v > 0$  and $f(x)$ are defined respectively as the coupling parameter and the potential shape.  The literature does provide explicit solved examples, such as the square-well potential \cite{ssw,Gr1}, the exponential potential \cite{bl, Gr2}, the Woods-Saxon potential, and the cusp potential \cite{vc}.  Based on these examples it is clear that the relation  $E(v)$ is not monotonic as it is in the Dirac relativistic equation, \cite{RH-99, RH-08, RP-15,R:P-16}, and indeed for Schr\"odinger's non-relativistic equation.  Consequently, comparison theorems for the Klein Gordon equation were restricted to positive energies \cite{RA-08, RH-10, RP-16}, and some are only valid for the ground state. In the present study, we have established new comparison theorems valid for negative potentials and for both positive and negative eigen-energies, and not just for the ground state. Throughout this paper, $V$ represents the time component of a four-vector; the scalar potential (a linear perturbation of the mass) is assumed to be zero.
 
 \medskip The idea that had a profound effect on the present work and, in particular, eliminated an earlier positivity restriction for energies, was our thinking of $v$ as a function of $E$.  This enabled us to arrive at a {\it function} $v(E)$, whereas $E(v)$ is a two-valued expression.  In section $2$, we establish the principal features of the spectral curves $v(E)$ for the class of negative bounded potentials that vanish at infinity. In section $3$ we  solve the Klein--Gordon equation analytically for the square-well potential in $d\geq 1$ dimensions. In section $4$ we prove some comparison theorems: the principal results claim that for any discrete eigenvalue $E \in  (-m,\, m)$ and negative potential-shape functions  $f_1(r)$ and $f_2(r)$ we have $f_1(r)\leq f_2(r)\implies v_1\leq v_2$.   In section $5$ we exhibit a complete recipe for spectral bounds for this class of potentials based on comparisons with the exactly soluble square-well problem.
\section{General features of the spectral curve $G(E) = v(E)$.}

\subsection{One-dimensional case}
The Klein--Gordon equation in one dimension is given by:
\begin{eqnarray}\label{KG}
\varphi^{\prime\prime}(x) =\big[m^2-\big(E-V(x)\big)^2\big]  \varphi(x),
\end{eqnarray}
 where $\varphi^{\prime\prime}$ denotes the second order derivative of $\varphi$ with respect to $x$, natural units $\hbar=c=1$ are used, and $E$ is the energy of a spinless particle of mass $m$.
 We suppose that the potential function $V$ is expressed as $V(x)=v f(x)$ with $v>0$ and $f$ satisfies the following conditions:
\begin{enumerate}
\item  $V(x) = v f(x), x\in\mathbb{R}$, where $v>0$ is the coupling parameter and $f(x)$ is the potential shape;
\item  $f$ is even $f(x) = f(-x)$;
\item  $f$ is not identically zero and non-positive;
\item  $f$ is attractive, that is $f$ is monotone non-decreasing over $[0,\infty)$;
\item  $f$ vanishes at infinity, i.e $\displaystyle\lim_{x\to \pm \infty} f(x)= 0$.
\end{enumerate}
We also assume that $V(x)=v f(x)$ is in this class of potentials, for which the Klein--Gordon equation (\ref{KG}) has at least one discrete eigenvalue $E$, and that equation (\ref{KG}) is the eigen-equation for the eigenstates.  Because of condition $5$, equation (\ref{KG})  has the asymptotic form
\begin{eqnarray*}
\varphi^{\prime\prime}=(m^2-E^2)\varphi,
\end{eqnarray*}
at infinity, with solutions $\varphi(x)=C_1 e^{\sqrt{k} |x|}+C_2 e^{-\sqrt{k} |x|}$, where $C_1$ and $C_2$ are constants of integration, and $k=m^2-E^2$.
The radial wave function of $\varphi$ vanishes at infinity; thus, $C_1=0$. Since $\varphi \in L^2(\mathbb{R})$, then $k>0$ which means that
\begin{eqnarray}\label{cond}
 |E|<m.
\end{eqnarray}
Suppose that $\varphi(x)$ is a solution of (\ref{KG}). Then by direct substitution we conclude that $\varphi(-x)$ is another solution of (\ref{KG}). Thus, by using linear combinations, we see that all the solutions of this equation may be assumed to be either even or odd.
Hence, if $\varphi$ is even then $\varphi^\prime(0) = 0$, and if $\varphi$ is odd then $\varphi(0) = 0$.
Since $\varphi\in L^2(\mathbb{R})$ then $\int_{-\infty}^{+\infty}\varphi^2 dx <\infty$. This means that the wave functions can be normalized and consequently we shall assume that $\varphi$ satisfies the normalization condition
\begin{eqnarray}\label{norm}
 ||\varphi||^2=\int_{-\infty}^{\infty} \varphi^2(x)  dx=1.
\end{eqnarray}
\begin{definition}
We denote by $\langle f\rangle$ and $\langle f^2\rangle$ the mean values of $f$ and $f^2$ respectively, where
$\langle\varphi,\psi\rangle = \int_{-\infty}^{\infty}\varphi(x)\psi(x) dx$ is the inner product on $L^2(\mathbb{\R})$,  that is    $\langle f\rangle=\langle\varphi,f\varphi\rangle=\int_{-\infty}^\infty f(x)\varphi^2(x)dx$ and      $\langle f^2\rangle=\langle\varphi,f^2\varphi\rangle=\int_{-\infty}^\infty f^2(x)\varphi^2(x) dx$.
\end{definition}
\medskip
\begin{lemma}
\begin{eqnarray}\label{mineq}
2E\langle f\rangle \,  <  \, v \langle f^2\rangle, \, \forall \, |E| \, <m.
\end{eqnarray}
\end{lemma}
\begin{proof} Expanding equation (\ref{KG}) we get:
\begin{eqnarray*}
\varphi^{\prime\prime}(x)=(m^2-E^2)\varphi(x) + 2Evf(x)\varphi(x) - v^2 f^2(x)\varphi(x).
\end{eqnarray*}
Multiplying both sides by $\varphi$ and integrating over $\mathbb{R}$ we obtain:
\begin{eqnarray*}
\int_{-\infty}^\infty \varphi^{\prime\prime}(x)\varphi(x) dx=m^2-E^2+2Ev\langle f\rangle-v^2\langle f^2\rangle.
\end{eqnarray*}
After applying integration by parts and using the fact that $\varphi$ vanishes at $\pm\infty$, the left-hand side of the last equation becomes $-\int_{-\infty}^\infty \big(\varphi^{\prime}(x)\big)^2dx$.
Thus, $-\int_{-\infty}^\infty \big(\varphi^{\prime}(x)\big)^2dx+E^2-m^2=2Ev\langle f\rangle-v^2\langle f^2\rangle$.
Since the left-hand side is negative, we have the desired result.
\end{proof}
\noindent We now define the operator $K$ as:
\begin{eqnarray}\label{op}
K=-\frac{\partial^2}{\partial x^2}+2Evf-v^2f^2.
\end{eqnarray}
If $\varphi$ is solution of the Klein--Gordon equation (\ref{KG}), then we have:
\begin{eqnarray}\label{KG2}
K\varphi = (E^2-m^2)\varphi,
\end{eqnarray}
and it follows 
\begin{eqnarray}\label{KG3}
\langle\varphi,K\varphi\rangle=\langle\varphi,(E^2-m^2)\varphi\rangle=E^2-m^2.
\end{eqnarray}
We observe that the domain of $K$ is $D_K = H^2(\mathbb{R})$, where $H^2(\mathbb{R})$ is the Sobolev space
defined as follows:
\begin{eqnarray*}
H^2(\mathbb{R}) = \{\varphi\in L^2(\mathbb{R}):\varphi^\prime, \varphi^{\prime\prime}\in L^2(\mathbb{R})\}.
\end{eqnarray*}
Since $||K\varphi|| = |E^2-m^2|\cdot ||\varphi||\leq m^2 ||\varphi||$ for all $\varphi\in D_K$, then $K$ is a bounded operator.
This implies that $K$ is continuous.
We also observe that $K$ is symmetric, that is to say: $\langle\varphi, K\psi\rangle = \langle K\varphi, \psi\rangle$.
\medskip
We now consider a family of Klein--Gordon spectral problems where $v = v(E)$ is a function of $E$. Let $\varphi_E$ denote the partial derivative of $\varphi$ with respect to $E$.
If we differentiate the normalization integral (\ref{norm}) partially with respect to $E$, we obtain the orthogonality relation 
$\langle\varphi,\varphi_E\rangle=0.$
Furthermore, differentiating equation (\ref{KG3}) with respect to $E$ we obtain:
\begin{eqnarray}\label{diff}
\langle\varphi_E,K\varphi\rangle+\langle\varphi,K_E\varphi\rangle+\langle\varphi,K\varphi_E\rangle=2E.
\end{eqnarray}
The symmetry of $K$ and the orthogonality of $\varphi$ and $\varphi_E$ imply that
\begin{eqnarray*}
 \langle\varphi,K\varphi_E\rangle=\langle\varphi_E,K\varphi\rangle=(E^2-m^2)\langle\varphi_E,\varphi\rangle=0.
\end{eqnarray*}
Then by using the expression 
\begin{eqnarray}\label{K_E}
K_E=2vf+2Evf-2vv_Ef^2
\end{eqnarray}
in equation (\ref{diff}), we obtain the key equation for our theorem in this section, namely:
\begin{eqnarray}\label{thm1}
v_E=\dfrac{E-v\langle f\rangle}{E\langle f\rangle-v\langle f^2\rangle}.
\end{eqnarray}
\medskip
\begin{thm}
If $E\geq v\langle f\rangle$, then $E_1\leq E_2\Rightarrow v(E_1)\geq v(E_2)$; and
 if $E < v\langle f\rangle$, then $E_1 <  E_2\Rightarrow v(E_1) < v(E_2)$.
\end{thm}
\begin{proof}
If $E\geq 0$, then $v_E\geq 0$, and the theorem holds immediately.
On the other hand, if $E < 0$, then by  (\ref{mineq}), $E\langle f\rangle < 2E\langle f\rangle < v\langle f^2\rangle$, which means that $E\langle f\rangle -  v\langle f^2\rangle <  0$.
Thus  $v\langle f\rangle \leq E<0 \Rightarrow v_E>0$, and  $E < v\langle f\rangle \Rightarrow v_E < 0.$
Therefore, the theorem has been proven.
\end{proof}
\begin{thm}
The spectral curve $G(E)$ is concave for all $E$, $|E| < m$
\end{thm}
\begin{proof}
Suppose that for any $|E_i|<m$, $\varphi_i$ is the corresponding wave function,\\
 $\langle f_i\rangle =\displaystyle\int_{-\infty}^{+\infty} f(x)\cdot(\varphi_i)^2 dx$, and $\langle f_i^2\rangle = \displaystyle\int_{-\infty}^{+\infty} f^2(x)\cdot(\varphi_i)^2 dx $.
Let the maximum value of $G(E)$ be equal to $v_{cr}$. By {\bf theorem II.1}, the corresponding value of $E$ is $E_{cr} = v_{cr}\langle f_{cr}\rangle$. Consider the point $A\big(E_{cr},  v_{cr}\langle f_{cr}\rangle$\big) and fix the point $B\big(E_n, v_n\big)$ on the spectral curve $G(E)$, where $v_n = G(E_n)$, such that $E_n\neq E_{cr}$ and $E_n\in(-m,m)$. Then
\begin{eqnarray*}
(AB): G_c(E) = v_{cr} + \dfrac{v_{cr} - v_n}{v_{cr}\langle f_{cr}\rangle - E_n}\big(E - v_{cr}\langle f\rangle\big).
\end{eqnarray*}
Assume, without loss of generality, that $E_{cr} < E_n$, and consider the function $t(E) = G(E) - G_c(E)$ where $E\in [E_{cr}, E_n].$
Then $t'(E) = \dfrac{dt}{dE} = \dfrac{E - v\langle f\rangle}{E\langle f\rangle - v\langle f^2\rangle} + \dfrac{v_n - v_{cr}}{v_{cr}\langle f_{cr}\rangle - E_n}$, which vanishes at the point $\big(E_m, G(E_m)\big)$ with
\begin{eqnarray*}
E_m = \dfrac{v_m E_n\langle f_m\rangle - v_m v_{cr}\langle f_m\rangle\langle f_{cr}\rangle + v_m v_{cr}\langle f^2_m\rangle - v_m v_n\langle f^2_m\rangle}{E_n - v_{cr}\langle f_{cr}\rangle + \langle f_m\rangle\big(v_{cr} - v_n\big)}\\
 = \dfrac{v_m \langle f_m\rangle\big(E_n - v_{cr}\langle f_{cr}\rangle\big) + v_m\langle f_m\rangle^2 (v_{cr} - v_n) + v_m\langle f_m^2\rangle (v_{cr} - v_n) - v_m\langle f_m\rangle^2 (v_{cr} - v_n)}{E_n - v_{cr}\langle f_{cr}\rangle + \langle f_m\rangle\big(v_{cr} - v_n\big)}.
\end{eqnarray*}
Hence
\begin{eqnarray}\label{E_m}
E_m = v\langle f_m\rangle + \dfrac{v_m\big(v_{cr} - v_n\big)\big(\langle f_m^2\rangle - \langle f_m\rangle^2\big)}{E_n - v_{cr}\langle f_{cr}\rangle + \langle f_m\rangle\big(v_{cr} - v_n\big)}.
\end{eqnarray}
Since $t'(E_{cr}) > 0$ and $t(E_{cr}) = t(E_n) = 0$, then the point $\big(E_m, G(E_m)\big)$ must be a maximum point of $t(E)$ over the interval $[E_{cr}, E_n]$ and $t(E) \geq 0$ over $[E_{cr}, E_n]$. This means that the chord $[AB]$ is always beneath the spectral curve $G(E)$ over $[E_{cr}, E_n]$ and the proof is complete.
\end{proof}
\noindent We observe that $\varphi$ is not necessarily required to be a node-free state, this means that the previous theorem is valid for all states, ground state and any excited state.
\subsection {d-dimensional cases ($d>1$):}
The Klein--Gordon equation in $d$ dimensions is given by
\begin{eqnarray*}
\Delta_d\Psi (r)=[m^2-(E- V(r))^2]\Psi (r),
\end{eqnarray*}
where natural units $\hbar=c=1$ are used and $E$ is the discrete energy eigenvalue of a spinless particle of mass $m$.
We suppose here that the vector potential function $V(r)$, $r=||\bf{r}||$, is a radially-symmetric Lorentz vector potential (the time component of a space-time vector), which satisfies the following conditions:
\begin{enumerate}
\item $V$ is not identically zero and non-positive, that is $V\leq 0$;
\item $V$ is attractive;
\item $V$ vanishes at $\infty$.
\end{enumerate}
The operator $\Delta_d$ is the $d$-dimensional Laplacian. Hence, the wave function for $d > 1$ can be expressed as $\Psi (r) =R(r) Y_{l_{d-1,...,l_1}} (\theta_1, \theta_2,...,\theta_{d-1})$, where $R\in L^2 (\mathbb{R}^d)$ is a radial function and $Y_{l_{d-1,...,l_1}}$ is a normalized hyper-spherical harmonic with eigenvalues $l(l+d-1)$, $l = 0, 1, 2, ...$ \cite{rtd}
The radial part of the above Klein--Gordon equation can be written as:
\begin{eqnarray*}
\frac{1}{r^{d-1}}\frac{\partial}{\partial r}\bigg(r^{d-1}\frac{\partial}{\partial r}R(r)\bigg) = \bigg[m^2-\big(E-V(r)\big)^2+\frac{l(l+d-2)}{r^2}\bigg]R(r),
\end{eqnarray*}
where $R$ satisfies the second-order linear differential equation
\begin{eqnarray*}
R^{\prime\prime}(r) +\frac{d-1}{r} R^{\prime} (r)= \bigg[m^2-\big(E-V(r)\big)^2+\frac{l(l+d-2)}{r^2}\bigg]R(r).
\end{eqnarray*}
Applying the change of variable $R(r) = r^{-\frac{d-1}{2}}\varphi(r)$, we obtain the following reduced second-order differential equation:
\begin{eqnarray}\label{KG4}
\varphi^{\prime\prime}(r)=\bigg[ m^2-\big(E-V(r)\big)^2+\frac{Q}{r^2}\bigg]  \varphi(r),
\end{eqnarray}
where 
\begin{eqnarray*}
Q=\frac{1}{4}(2l+d-1)(2l+d-3),
\end{eqnarray*}
with $l = 0, 1, 2, ...$ and $d = 2, 3, 4, ...,$ which is the radial Klein-Gordon equation for $d > 1$ dimensions. The reduced wave function $\varphi$ satisfies $\varphi(0) = 0$ and
$\displaystyle\lim_{r\to\infty}\varphi = 0$ \cite{MMN}. For bound states, the normalization condition is:
\begin{eqnarray*}
\int_{0}^{\infty}\varphi^2 (r) dr= 1.
\end{eqnarray*}
Since $V$ vanishes at $\infty$, then equation (\ref{KG4}) becomes 
\begin{eqnarray*}
\varphi^{\prime\prime}=(m^2-E^2)\varphi
\end{eqnarray*}
near infinity, which means that $|E| <  m$ by the same reasoning as used for equation (\ref{cond}).
Since the derivative of the term $\frac{Q}{r^2}$ in equation (\ref{KG4}) with respect to $E$ is equal to zero, then  by the same reasoning the relation (\ref{mineq}) is also valid for all other $d>1$ dimensions.
We define the operator 
\begin{eqnarray*}
K=-\frac{\partial^2}{\partial r^2} + 2Evf - v^2f^2 + \frac{Q}{r^2}.
\end{eqnarray*}
We have
\begin{eqnarray*}
\langle K\rangle =\langle (E^2 - m^2)\varphi,\varphi\rangle = E^2-m^2.
\end{eqnarray*}
As in the previous section, $K$ is bounded and symmetric with $D_K = H^2(\mathbb{R}^d)$.
We consider the same family of Klein--Gordon spectral problems with $v = v(E)$
Since
$ K_E = 2vf + 2Ev_E f - 2vv_E f^2,$
 then we obtain the same relation as in the one-dimensional case,
\begin{eqnarray*}
v_E=\frac{E-v\langle f\rangle}{E\langle f\rangle-v\langle f^2\rangle}.
\end{eqnarray*}
\medskip
\begin{thm}
If $E\geq v\langle f\rangle$, then $E_1\leq E_2\Rightarrow v(E_1)\leq v(E_2)$, and
if $E < v\langle f\rangle$, then $E_1 <  E_2\Rightarrow v(E_1) > v(E_2)$.
\end{thm}
\begin{proof}
Same as {\bf Theorem II.1}.
\end{proof}
\begin{thm}
The spectral curve $G(E)$ is concave for all $|E| < m$
\end{thm}
\begin{proof}
Same as {\bf Theorem II.2}
\end{proof}
\noindent We observe that as in the one-dimensional case, this theorem does not require the radial wave function to be node-free; it's valid for both ground and excited states.

\begin{figure}[htbp]
\centering
\includegraphics[scale=0.35]{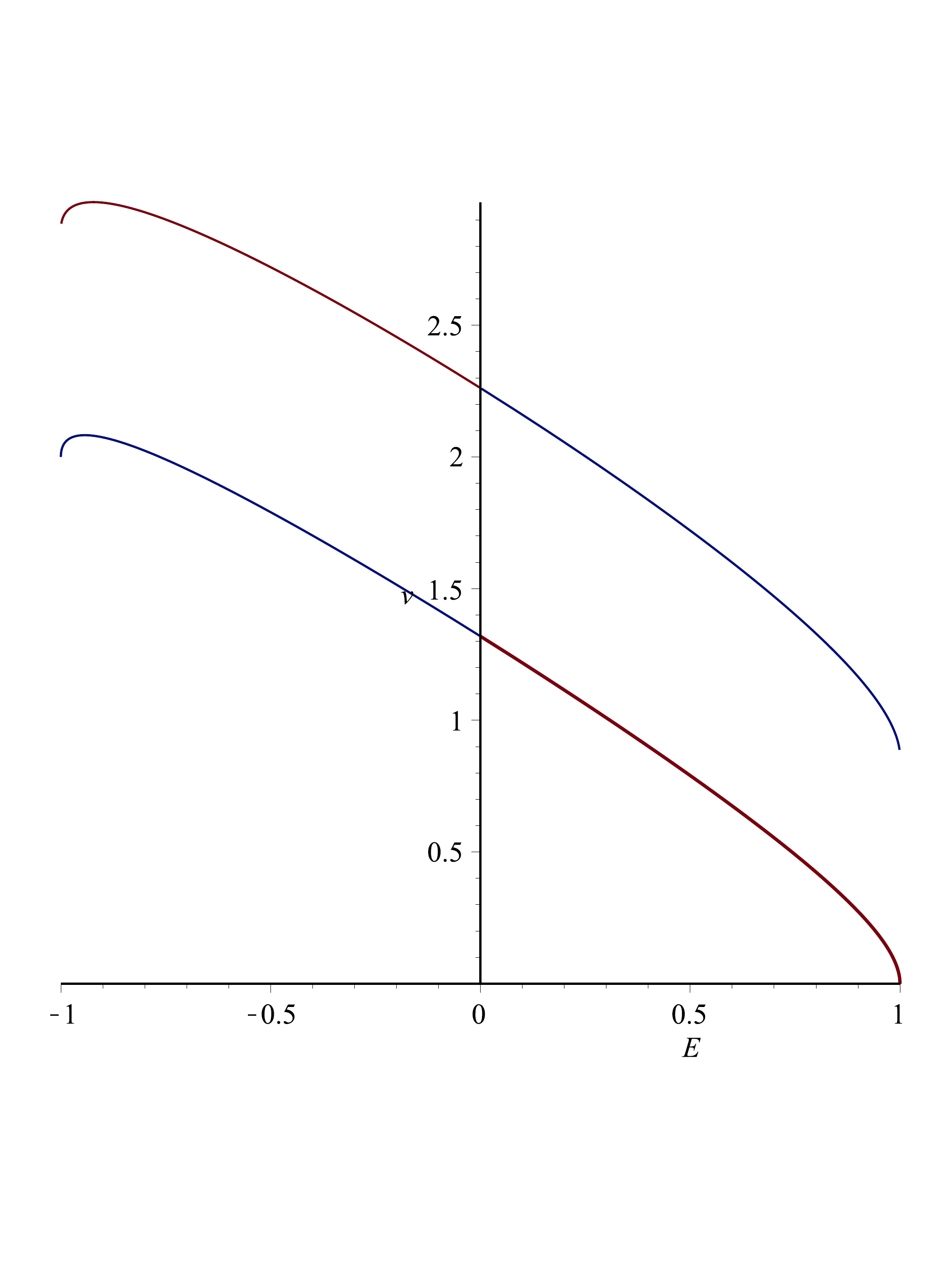}
\caption{Graphs of the function $v(E)$ in the square-well potential, for the ground state and the first excited state. We exhibit an explicit expression in section 5. }
\end{figure}

\begin{figure}
\centering
\includegraphics[scale=0.35]{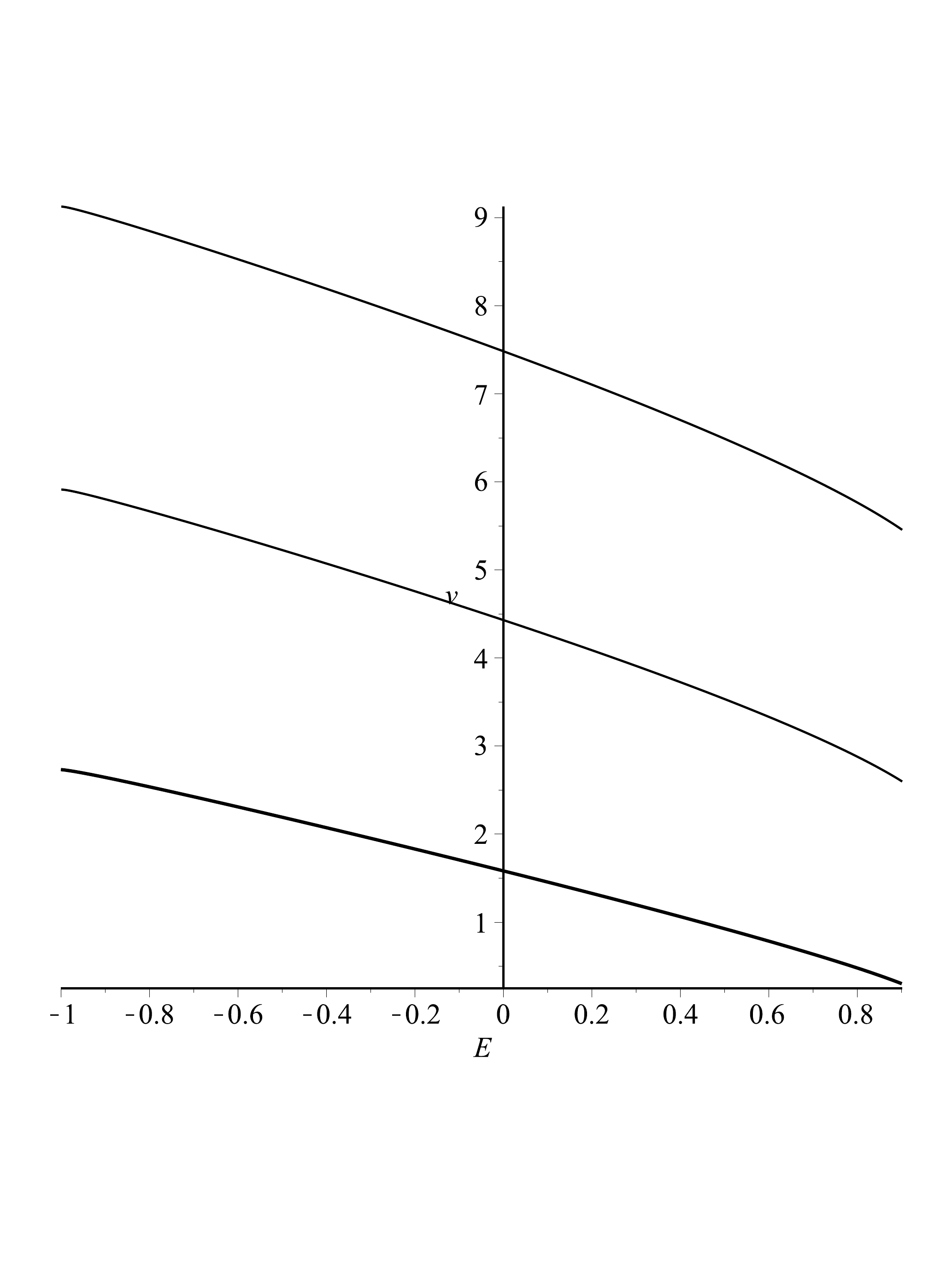}
\caption{Graphs of the function $v(E)$ in the Woods-Saxon potential shape $f(x) = -\dfrac{1}{1 + e^{b(|x|-1)}}$ with $b = \frac{20}{7}$, for the ground state, the first excited state, and the second excited state. The graphs are plotted using a numerical shooting method of our own.}
\end{figure}
\newpage
\section{Exact solution for the Klein--Gordon equation with the square-well potential}
\begin{enumerate}
\item {\bf One dimensional case}:
Consider the Klein--Gordon equation in dimension $d = 1$:     
$\varphi^{\prime\prime} (x)= [m^2 - \big(E - g(x,t)\big)^2]\varphi(x)$, and the square-well potential
\[ g(r,t) = 
\begin{cases}
           -v_0, &  |x| \leq t  \\  
           0, & {\rm elsewhere}
      \end{cases}
,\]
where $v_0>0$. For $x < -t$, we get: $\varphi^{\prime\prime}(x) = (m^2 - E^2)\varphi(x)$. Thus, $\varphi (x)= Ae^{-kx} + Be^{kx}$ with $k^2 = m^2 - E^2$. Since $\varphi$ vanishes at $-\infty$, then $A = 0$ and $\varphi (x)=  Be^{kx}$. Similarly, for $x > t$ we obtain $\varphi (x)=  Ce^{-kx}$. For $|x|\leq t$, $\varphi^{\prime\prime}(x) + w^2\varphi (x)= 0$ with $w = \sqrt{(E+v_0)^2 - m^2}$. Then $\varphi(x) = D \sin (wx) + E \cos (wx).$ Since,as shown in section II-A, all the solutions are either even or odd, then the even solution is 
\[ \varphi(x) = 
\begin{cases}
            Be^{kx}, & x < -t \\    
            E\cos(wx), &  |x| \leq t  \\  
            Ce^{-kx}, & x > t
      \end{cases}
,\]
and the odd solution is
\[ \varphi(x) = 
\begin{cases}
            Be^{kx}, & x < -t  \\   
            D\sin(wx), &  |x| \leq t  \\  
            Ce^{-kx}, & x > t
      \end{cases}
.\]
Regarding the even solution, since $\varphi$ is required  to be continuously differentiable at $t$, then
\begin{eqnarray}\label{EQC}
E\cos (wt) = Ce^{-kt},
\end{eqnarray}
and
\begin{eqnarray}\label{EQD}
-Ew\sin (wt) = -Cke^{-kt}.
\end{eqnarray}
Dividing equation (\ref{EQC}) by (\ref{EQD}), we obtain the eigenvalue equation
\begin{eqnarray}\label{SolE}
w\tan{wt} = k.
\end{eqnarray}
Similarly, the eigenvalue equation for the odd states reads
\begin{eqnarray}\label{SolO}
w\cot(wt) = -k.
\end{eqnarray} 
These equations allow us to compute the eigenvalue $v_0$ given the energy $E$.
\item {\bf $d>1$ dimensional cases}:
The radial part of the Klein--Gordon equation reads \cite{rtd}    
\begin{eqnarray}\label{kd>1}
R^{\prime\prime}(r) +\frac{d-1}{r} R^{\prime} (r)= \bigg[m^2-\big(E-g(r,t)\big)^2+\frac{l(l+d-2)}{r^2}\bigg]R(r),
\end{eqnarray}
 where $l = 0, 1, 2,...$ and 
\[ g(r,t) = 
\begin{cases}
           -v_0, &  r \leq t   \\ 
           0, & {\rm elsewhere}
      \end{cases}
\] 
with $v_0 > 0$. For $d = 3$, the eigenvalue equation is \cite{grnr}:
\begin{eqnarray}\label{sphbessel}
\frac{j^\prime_l(k_it)}{j_l(k_it)} = \frac{h^{(1)\prime}_l(ikt)}{h^{(1)}_l(ikt)},
\end{eqnarray}
where $k_i^2 = (E + v_0)^2 - m^2$, $k^2 = m^2 - E^2$, $i^2 = -1$, $j_l$ is the spherical Bessel function of the first kind, and $h_l^{(1)}$ is the Hankel function of the first kind. In particular, the eigenvalue equation for the $s$-states $(l = 0)$ is \cite{grnr}:
\begin{eqnarray*}
k_i\cot(k_it) = -k.
\end{eqnarray*}
To generalize for any $d$-dimensional case, we consider the reduced form of the radial part of the Klein--Gordon equation (\ref{KG4}). For $r < t$, we write it as 
\begin{eqnarray*}
r^2\varphi^{\prime\prime} + [k_i^2r^2 - Q]\varphi(r) = 0.
\end{eqnarray*}
 Changing the variable $r$ into $\sigma = k_ir$ we obtain the following differential equation:
\begin{eqnarray*}
\sigma^2\varphi^{\prime\prime}(\sigma) + [\sigma^2 - \nu(\nu + 1)]\varphi(\sigma) = 0,
\end{eqnarray*}
where $\nu = \frac{2l + d - 3}{2}$. This is the Ricatti-Bessel equation with solution 
$\varphi(\sigma) = C_1\sigma j_{\nu}(\sigma) + C_2\sigma y_{\nu}(\sigma)$ \cite{abw}, where $y_{\nu}$ is the spherical Bessel function of the second kind. Since we have an irregular point at $\sigma = 0$, then $C_2 = 0$. For $r > t$ we obtain the differential equation
\begin{eqnarray*}
r^2\varphi^{\prime\prime}(r) + [-k^2r^2 - Q]\varphi(r) = 0.
\end{eqnarray*}
Using the change of variable $\sigma = ikr$ we obtain
\begin{eqnarray*}
\sigma^2\varphi^{\prime\prime}(\sigma) + [\sigma^2 - \nu(\nu + 1)]\varphi(\sigma) = 0,
\end{eqnarray*}
whose general solution is \cite{abw} $\varphi(\sigma) = W_1\sigma h^{(1)}(\sigma) + W_2\sigma h^{(2)}(\sigma)$, where $h^{(2)}$ is the Hankel function of the second kind. Since $\varphi\in L^2(\mathbb{R})$, then $W_2 = 0$. Since $\varphi$ is continuously differentiable at $r = t$, then the corresponding eigenvalue equation is  
\begin{eqnarray}\label{ricbessel}
\frac{[(k_it) j_l(k_it)]^{\prime}}{(k_it) j_l(k_it)} = \frac{[(ikt)h^{(1)}_l(ikt)]^{\prime}}{(ikt) h^{(1)}_l(ikt)}.
\end{eqnarray}
\end{enumerate}

\begin{center}
\section{Comparison theorems for pairs of potential functions with different potential shapes}
\end{center}
\subsection{$d=1$~dimensional case}
Consider the Klein--Gordon equation in one dimension
\begin{eqnarray}\label{KKG}
\varphi^{\prime\prime}(x)=[m^2-(E-V(x))^2]\varphi(x),
\end{eqnarray}
where natural units $\hbar = c = 1$ are used, and $E$ is the energy of a spinless particle of mass $m$. We assume that $V(x)=v f(x)$ with the same conditions in {\bf Section 0.2.1}, that is:
\begin{enumerate}
\item  $V = v f$, where $v>0$ is the coupling parameter and $f$ is the potential shape of $V$;
\item  $V$ is an even function, that is $V(x) = V(-x)$;
\item  $V$ is not identically zero and a non-positive function, i.e $V\leq 0$;
\item  $f$ is attractive, that is $f$ is monotone non-decreasing over $[0,\infty)$;
\item  $f$ vanishes at infinity, i.e $\displaystyle\lim_{x\to \pm \infty} f(x)= 0$
\end{enumerate}
By similar reasoning in section {\bf Section 0.2.1}, we have $|E| < m$, and all the solutions of equation (\ref{KKG}) are either even or odd functions. 
We also assume that the wave function in this section satisfies the normalization condition, i.e,
\begin{eqnarray}\label{norm1}
||\varphi||^2 = \int_{-\infty}^{\infty}\varphi^2(x) dx = 1.
\end{eqnarray}
In this section, we consider the parameter $a\in [0,1]$ and the two potential shapes $f_1$ and $f_2$ with $f = f(a,x) = f_1(x) + a\big[f_2(x) - f_1(x)\big]$, where $f_1\leq f_2 \leq 0$.
Hence $f\leq 0$, attractive, even, vanishes at infinity, $f(0,x) = f_1(x)$ when $a=0$, and  $a = 1$ when $f(1,x) = f_2(x)$,
and
\begin{eqnarray}\label{rel}
\frac{\partial f}{\partial a} = f_2(x) -f_1(x)\geq 0.
\end{eqnarray}
 Hence, $f$ is monotone non-decreasing in the parameter $a$.
The idea in this section is to study the variations of the coupling $v$ with respect to $a$, provided $v = v(a)$ and the value of $E$ is given as a constant, that is $\frac{\partial E}{\partial a} = 0$, and $-m < E < m$.
We again consider the symmetric bounded operator $K$ in (\ref{op}), and we define $\varphi_a$ to be the partial derivative of $\varphi$ with respect to $a$.
Differentiating equation (\ref{KG3}) with respect to $a$ we get:
\begin{eqnarray}\label{diff1}
\langle\varphi_a,K\varphi\rangle + \langle\varphi,K_a\varphi\rangle + \langle\varphi,K\varphi_a\rangle = 0
\end{eqnarray}
Applying the partial derivative with respect to $a$ to equation (\ref{norm1}) and using the symmetry of $K$, we obtain the new orthogonality relation
\begin{eqnarray*}
\langle\varphi_a,K\varphi\rangle = \langle\varphi, K\varphi_a\rangle = (E^2-m^2)\langle\varphi_a,\varphi\rangle = 0.
\end{eqnarray*}
We also have:
\begin{eqnarray*}
K_a = 2Ev_af + 2Ev(f_2 - f_1) - 2vv_af^2 - 2v^2f(f_2-f_1),
\end{eqnarray*}
with $v_a$ defined as $\frac{\partial v}{\partial a}$.
Equation (\ref{diff1}) becomes:
\begin{eqnarray*}
Ev_a\langle f\rangle + Ev\int_{-\infty}^{\infty}\bigg(f_2(x) - f_1(x)\bigg)\varphi^2(x)dx - vv_a\langle f^2\rangle - v^2\int_{-\infty}^{\infty}f\bigg(f_2(x) - f_1(x)\bigg)\varphi^2(x) dx= 0.
\end{eqnarray*}
This leads us to the following relation:
\begin{eqnarray}\label{rel1}
v_a = \dfrac{v I}{E\langle f\rangle - v\langle f^2\rangle},
\end{eqnarray}
where
\begin{eqnarray}\label{I}
I =\displaystyle\int_{-\infty}^{\infty}\bigg(f_2(x) - f_1(x)\bigg) \bigg(vf(x) - E\bigg)\varphi^2(x) dx.
\end{eqnarray}
 In the next two lemmas, we shall use the parity of $\varphi$ and study the sign of $\varphi^{\prime\prime}$ on the interval $[0,\infty)$.
\medskip
\begin{lemma}
If $\varphi$ is the node-free (ground) state, then $\varphi^{\prime\prime}$ changes its sign only once over $[0,\infty)$.
\end{lemma}
\begin{proof}
Let $\varphi^{\prime\prime}(x) = 0$. Then from equation (\ref{KKG}) we get $m^2 - \big(E - V(x)\big)^2 = 0$, which means that $V(x) = E - m$ or $V(x) = E + m$.
Since $|E| < m$ and $V(x)\leq 0$, then $V\neq E+m$. Hence, $V(x) = E - m$ and $\varphi^{\prime\prime}(x) = 0 \Longleftrightarrow x =V^{-1}( E - m)$, where $V^{-1}$ is the inverse of the monotone function $V$.
\begin{enumerate}
\item {\bf $V$ is unbounded near $0$:}
Since $V$ is unbounded near $0$, then $\varphi^{\prime\prime} <0$ near $0$, and since $V$ vanishes at $\infty$, equation (\ref{KKG}) becomes $\varphi^{\prime\prime} = (m^2 - E^2)\varphi > 0$.
Hence,  $\varphi$ is concave on $\big[0,V^{-1}(E - m)\big)$ and convex on $\big(V^{-1}(E - m),\infty\big)$.
\item {\bf $V$ is bounded; that is: $V_0\leq V\leq 0$:}
    
Since $\varphi$ is an even state, then $\varphi^{\prime}(0) = 0$, which means that $y = \varphi(0)$ is an equation of the tangent line to $\varphi$ at $x = 0$.
    
If $\varphi$ is convex near $0$, then $\varphi^{\prime\prime}$ must change it sign at some $x_1\in [0,\infty)$ since we know that $\varphi$ vanishes near $\infty$. However, equation (\ref{KKG}) becomes $\varphi^{\prime\prime} = (m^2 - E^2)\varphi > 0$ near $\infty$, which means that $\varphi$ is convex near $\infty$. Thus $\varphi^{\prime\prime}$ should again change its sign at some $x_2\in [x_1,\infty).$
This means that $\varphi$ has two inflection points on $[0,\infty)$, which is a contradiction. 
Hence, $\varphi$ is concave on $\big[0,V^{-1}(E - m)\big)$ and convex on $\big(V^{-1}(E - m),\infty\big)$.
\end{enumerate}
\end{proof}
\medskip
\begin{lemma}
$\varphi^{\prime\prime}$changes its sign at least once over $[0,\infty)$, for any excited state $\varphi$.
\end{lemma}
 \begin{proof}Using the parity of $\varphi$, it is sufficient to study the sign of $\varphi^{\prime\prime}$ on the interval $[0,\infty)$.
If $V$ is unbounded near $0$, then $\varphi^{\prime\prime} < 0$ near $0$ and $\varphi^{\prime\prime} > 0$ near $\infty$.
If $V$ is bounded; that is $V_0\leq V\leq 0$ where $V_0 = V(0)$, then we divide the proof into the following two cases:
\begin{enumerate}
\item {\bf $\varphi$ has only one node:}
Suppose that $\varphi$ has one node $\alpha$, then $\varphi^{\prime\prime}(x) = 0\Longleftrightarrow x = \alpha$ or $x = V^{-1}(E - m)$.
If $\varphi(x) > 0$ for $x > \alpha$, then $\varphi$ should attain a maximum value since it vanishes near $\infty$, and thus $\varphi^{\prime\prime} < 0$. However, by the same condition that $\varphi$ vanishes near $\infty$, $\varphi^{\prime\prime}$ should change its sign one more time. This means that $V^{-1} (E - m)\in\big(\alpha,\infty)$.
Therefore
\begin{enumerate}
\item if $\varphi(x) > 0$ for $x >\alpha$, then $\varphi^{\prime\prime}(x) < 0$ for $x\in \big(\alpha, V^{-1}(E - m)\big)$, and $\varphi^{\prime\prime}(x) > 0$ for $x\in(0,\alpha)\cup \big(V^{-1}(E - m),\infty\big);$
\item  if $\varphi(x) < 0$ for $x>\alpha$, then $\varphi^{\prime\prime}(x) > 0$ for $x\in \big(\alpha, V^{-1}(E - m)\big)$, and $\varphi^{\prime\prime}(x) < 0$ for $x\in(0,\alpha)\cup \big(V^{-1}(E - m),\infty\big).$
\end{enumerate}
\item{\bf $\varphi$ has $n$ nodes, $n\geq 2$:}
Suppose that $\varphi$ has $n$ nodes, $x= \alpha_1, \alpha_2, ..., \alpha_n$, $n\geq 2$ . Then
\begin{center}
$\varphi^{\prime\prime}(x) = 0\Longleftrightarrow m^2 - \big(E - V(x)\big)^2 = 0$ or $\varphi (x)= 0,$
\end{center}
which means that
\begin{center}
$x = \alpha_1, \alpha_2, ..., \alpha_n$ or $V^{-1}(E - m).$
\end{center}
We shall now study the concavity of $\varphi$ over the interval $(\alpha_{n-1},\infty):$
If $\varphi(x) > 0$ on $\big(\alpha_{n-1},\alpha_n\big)$, then $\varphi$ must attain a maximum value at some $x_0\in\big(\alpha_{n-1},\alpha_n\big)$ and $\varphi$ is concave on $\big(\alpha_{n-1},\alpha_n\big)$.
 For $x >\alpha_n$, $\varphi$ changes both its sign and concavity. Thus $\varphi$ becomes convex and negative for $x >\alpha_n$.
However, since $\alpha$ vanishes near $\infty$, then $\varphi^{\prime\prime}$ vanishes and changes its sign one more time somewhere after its last node.
 This implies that $V^{-1}(E - m)\in\big(\alpha_n,\infty\big)$, and therefore $\varphi^{\prime\prime}(x) < 0$ for $x\in\big(\alpha_{n-1},\alpha_n\big)\cup\big(V^{-1}(E - m),\infty\big)$, and $\varphi^{\prime\prime}(x) > 0$ for $x\in\big(\alpha_n,V^{-1}(E - m)\big)$.
By the same reasoning, if $\varphi(x)<0$ on $(\alpha_{n-1},\alpha_n)$, then $\varphi^{\prime\prime} (x)> 0$ for $x\in\big(\alpha_{n-1},\alpha_n\big)\cup\big(V^{-1}(E - m),\infty\big)$, and $\varphi^{\prime\prime}(x) < 0$ for $x\in\big(\alpha_n,V^{-1}(E - m)\big)$.
\end{enumerate}
\end{proof}
\medskip
\begin{lemma}
The integral $I$ defined in relation (\ref{I}) is non-positive for any state $\varphi$.
\end{lemma}
\begin{proof}
We first write equation (\ref{KKG}) as
\begin{eqnarray*}
\big(\varphi(x)\big)E^2 - \big(2vf(x)\varphi(x)\big) E + \big(\varphi^{\prime\prime}(x) - m^2\varphi(x) + v^2f^2(x)\varphi(x)\big )= 0.
\end{eqnarray*}
This is a quadratic equation in $E$, and we have
\begin{center}
$E = vf(x)\pm\dfrac{\sqrt{v^2f^2(x)\varphi^2(x) -\big (\varphi(x)\varphi^{\prime\prime}(x) - m^2\varphi^2(x) + v^2f^2(x)\varphi^2(x)\big)}}{\varphi(x)}$.
\end{center}
Then
\begin{eqnarray}\label{choice 1}
 E = vf(x)- \sqrt{m^2 - \dfrac{\varphi^{\prime\prime}(x)}{\varphi(x)}},
\end{eqnarray}
or
\begin{eqnarray}\label{choice 2}
 E = vf(x)+ \sqrt{m^2 - \dfrac{\varphi^{\prime\prime}(x)}{\varphi(x)}},
\end{eqnarray}
In solution (\ref{choice 1}), $\varphi^{\prime\prime}$ cannot change its sign because if $\varphi^{\prime\prime}(x) < 0$, then 
\begin{eqnarray*}
 vf(x)- \sqrt{m^2 - \dfrac{\varphi^{\prime\prime}(x)}{\varphi(x)}}\leq - \sqrt{m^2 - \dfrac{\varphi^{\prime\prime}(x)}{\varphi(x)}} <-m,
\end{eqnarray*}
and we already know that $|E| < m$.
Hence, since we have shown in {\bf Lemma 0.3.2} that $\varphi^{\prime\prime}$ changes its sign, then $E$ can only take the second solution (\ref{choice 2}).Therefore, the relation (\ref{I}) becomes:
\begin{eqnarray*}
I = \displaystyle\int_{-\infty}^{\infty} -\big(f_2(x) - f_1(x)\big)\sqrt{m^2 - \dfrac{\varphi^{\prime\prime}(x)}{\varphi(x)}}\varphi^2(x) dx\leq 0.
\end{eqnarray*}
\end{proof}
\begin{thm}
\begin{eqnarray*}
f_1(x)\leq f_2(x)\Rightarrow v_1\leq v_2,
\end{eqnarray*}
for all $x\in[0,\infty)$.
\end{thm}
\begin{proof}

Consider the relation (\ref{rel1}).  If $E\geq 0$, then $E\langle f\rangle - v\langle f^2\rangle < 0$, and if $E < 0$,
 then using the relation (\ref{mineq}) we also get the same result.
Thus, the denominator of equation (\ref{rel1}) is negative for all $|E| < m$.   
Since we also proved in {\bf Lemma 0.3.3} that $I\leq 0$, then $v_a\geq 0$ for all $a\in [0,1]$ and $E\in(-m , m)$.    
This result completes the proof of the theorem.
\end{proof}
\subsection{d-dimensional cases ($d > 1$)}
In this section, we use the same reduced Klein--Gordon equation stated in (\ref{KG4}), with $\varphi$ satisfying $\varphi (0) = 0$ and the 
same normalization condition $\int_{0}^{\infty}\varphi^2(r) dr= 1$.    
We assume the same conditions for the potential shape $f$ as in {\bf section 0.3.1}.   
This proof is not valid for the s-states of the $2$-dimensional case, that is to say for $d = 2$ and $l = 0$. 
We shall prove this in the next section.
\begin{lemma}
$\varphi^{\prime\prime}$ changes its sign at least once, for any state $\varphi$.
\end{lemma}
\begin{proof}:
\begin{enumerate}
\item{\bf $\varphi$ is a node-free state:}  $\varphi^{\prime\prime}(r) = 0\Longleftrightarrow m^2 - \big(E-V(r)\big)^2 + \frac{Q}{r^2} = 0$, and near $\infty$, $\varphi^{\prime\prime} = (m^2 - E^2)\varphi > 0$, which means that $\varphi$ is convex near $\infty$.    
If $\varphi$ is concave near $0$, then the theorem is proved.   
If $\varphi$ is convex near $0$, then $\varphi^{\prime\prime}$ should change its sign at least at some  solutions $\{r_1, r_2\}$, of the equation $m^2 - \big((E - V(r)\big)^2 + \frac{Q}{r^2} = 0,$ in order to be positive near $\infty$.
\item {\bf $\varphi$ has one node:}  
Suppose that $\varphi$ has one node $\alpha$, then $\varphi^{\prime\prime} (r)= 0\Longleftrightarrow r = \alpha$ or $r = r_1, r_2, ..., r_n$, where the $r_i's$ are roots of the equation $ m^2 - \big(E - V(r)\big)^2 + \frac{Q}{r^2} = 0$, $i = 1 ... n$.
We now study the sign of $\varphi^{\prime\prime}$ for $r > \alpha$: If $\varphi (r)> 0$, then, owing to to the fact that $\varphi$ vanishes at $\infty$, it should attain a maximum value over the interval $(\alpha, \infty)$ becoming concave near $\alpha^+$. Similarly, we deduce that $\varphi^{\prime\prime}(r)$ should change its sign at least once over $(\alpha,\infty)$, 
implying that  there exists $r_i \in(\alpha, \infty)$.
If $\varphi (r)< 0$, then we can also prove this lemma by the same reasoning.
\item {\bf $\varphi$ has $n$ nodes, $n\geq 2$:}
    
Suppose that $\varphi$ has $n$ nodes, $\alpha_1, \alpha_2, ... , \alpha_n$ with $n\geq 2$.

Then $\varphi^{\prime\prime}(r) = 0\Longleftrightarrow m^2 -\big(E - V(r)\big)^2 +\frac{Q}{r^2} = 0$ or $\varphi = 0,$ which means:
    
$r = \alpha_1, \alpha_2, ..., \alpha_n, r_1, r_2, ..., r_n,$ where the $r_i's$ are the solutions of the equation
    
 $ m^2 - \big(E - V(r)\big)^2 +\frac{Q}{r^2} = 0$, $i = 1, 2, ..., n.$
    
We study the concavity of $\varphi$ over the interval $(\alpha_{n-1}, \alpha_n)$:
    
If there exists some $r_i\in (\alpha_{n-1}, \alpha_n)$, then $\varphi^{\prime\prime}$ changes its sign at least once over $(\alpha_{n-1},\alpha_{n})$
    
If there isn't any inflection point of $\alpha$ between $\alpha_{n-1}$ and  $\alpha_n$, then $\varphi^{\prime\prime}$ doesn't change its sign on $(\alpha_{n-1}, \alpha_n)$; however, since $\varphi$ vanishes at $\infty$, then there must be at least one inflection point $r_i\in  (\alpha_{n-1}, \alpha_n)$, which means that $\varphi$ changes its concavity at least once over $(\varphi_{n-1},\infty)$.
\end{enumerate}\end{proof}
\begin{lemma}
The integral $I$ defined in relation (\ref{I}) is non-positive for any state $\varphi$ and for all $d > 1$, except for the $s$-states of $d = 2$, that is: when $d = 2$ and $l = 0$.
\end{lemma}
\begin{proof}
The expression (\ref{KKG}), written as
\begin{eqnarray*}
\big(\varphi(r)\big)E^2 - \big(2vf(r)\varphi(r)\big) E + \bigg(\varphi^{\prime\prime}(r) - m^2\varphi(r) + v^2f^2(r)\varphi(r) -\frac{Q}{r^2}\varphi(r)\bigg) = 0,
\end{eqnarray*}
is a quadratic equation in $E$.
    
Thus, 
\begin{eqnarray*}
E = vf(r)\pm\sqrt{m^2 + \frac{Q}{r^2} - \dfrac{\varphi^{\prime\prime}(r)}{\varphi(r)}}.
\end{eqnarray*}
If $\varphi^{\prime\prime}(r) < 0$, then  $vf(r) - \sqrt{m^2 + \dfrac{Q}{r^2} - \frac{\varphi^{\prime\prime}(r)}{\varphi(r)}} < -m$, which means that $E$ cannot take this value since $|E| < m$. Hence,
\begin{eqnarray}\label{E2}
E =  vf(r) + \sqrt{m^2 + \dfrac{Q}{r^2} - \dfrac{\varphi^{\prime\prime}(r)}{\varphi(r)}}.
\end{eqnarray}
Using relation (\ref{E2}) in (\ref{I}) we get:
\begin{eqnarray*}
I = -\displaystyle\int_{0}^{\infty}\big(f_2(r) - f_1(r)\big)\sqrt{m^2 + \frac{Q}{r^2} - \frac{\varphi^{\prime\prime}(r)}{\varphi(r)}}\varphi^2(r) dr \leq 0.
\end{eqnarray*}
\end{proof}
\begin{thm}
\begin{eqnarray*}
f_1(r)\leq f_2(r)\Longrightarrow v_1\leq v_2,
\end{eqnarray*}
for all $r\in[0,\infty)$ and $d > 1$, except for the $s$-states for $d = 2$, that is, when $d = 2$ and $l = 0$.
\end{thm}
\begin{proof}
Same proof as {\bf theorem 0.3.1}
\end{proof}
\subsection{S-States for the $2$-dimensional case}
The reduced Klein--Gordon equation in this case reads
\begin{eqnarray*}
\varphi^{\prime\prime} (r) = \bigg[m^2 - \big(E - V(r)^2\big)-\dfrac{1}{r^2}\bigg]\varphi(r).
\end{eqnarray*}
Thus $E = vf(r)\pm\sqrt{m^2 - \dfrac{1}{r^2}-\dfrac{\varphi^{\prime\prime}(r)}{\varphi(r)}}.$
Eliminating the solution $E = vf(r) - \sqrt{m^2 - \dfrac{1}{r^2}-\dfrac{\varphi^{\prime\prime}(r)}{\varphi(r)}}$ fails because of the existence of the term $-\dfrac{1}{r^2}$, and consequently, the proof of {\bf theorem 0.3.1} is not valid.
Thus, we use the non-reduced form of the Klein--Gordon radial equation, namely 
\begin{eqnarray*}
R^{\prime\prime}(r) + \dfrac{d-1}{r} R^\prime (r) =\bigg[m^2 - \big(E-V(r)\big)^2+\dfrac{l(l+d-2)}{r^2}\bigg]R(r),
\end{eqnarray*}
where $d = 2$, $l = 0$, and $\displaystyle\int_0^\infty R^2(r) dr = 1$. Hence,
\begin{eqnarray}\label{NR}
R^{\prime\prime}(r) + \frac{1}{r}R^\prime(r) = \bigg[m^2 - \big(E - V(r)\big)^2\bigg]R(r).
\end{eqnarray}
We assume that $V = vf$, with $f$ satisfying the same conditions of {\bf section 0.3.1}.
    
Define the symmetric operator
\begin{eqnarray*}
K = -\frac{\partial^2}{\partial r^2} - \frac{\partial}{\partial r} + 2Evf - v^2 f^2.
\end{eqnarray*}
Then
\begin{eqnarray}\label{exp}
\langle K\rangle = E^2 - m^2.
\end{eqnarray}
    
Differentiating (\ref{exp}) with respect to the parameter $a$ we get
\begin{eqnarray}\label{df}
 \langle R_a,KR\rangle + \langle R,K_aR\rangle + \langle R, KR_a\rangle = 0,
\end{eqnarray}
where $K_a = \frac{\partial K}{\partial a}$.
    
But 
\begin{eqnarray*}
\frac{\partial}{\partial a}\bigg[\displaystyle\int_0^\infty R^2(r) dr\bigg] = 2 \displaystyle\int_0^\infty R(r)\frac{\partial R(r)}{\partial a} = 0.
\end{eqnarray*}
Then we obtain the orthogonality relation $\langle R,R_a\rangle = \langle R_a, R\rangle = 0$.
    
Therefore, $\langle R_a,KR\rangle = \langle R, KR_a\rangle = (E^2-m^2)\langle R,R_a\rangle = 0$, with $R_a = \frac{\partial R}{\partial a}$.
We also have $K_a = 2Ev_af + 2Ev(f_2 - f_1) - 2vv_af^2 - 2v^2f(f_2 - f_1),$ where $v_a = \frac{\partial v}{\partial a}$.
    
Thus, using $K_a$ in equation (\ref{df}) we obtain
\begin{eqnarray}\label{va}
v_a = \dfrac{v\bigg[\displaystyle\int_0^\infty\big(f_2(r) - f_1(r)\big)\big(vf(r) - E\big)R^2(r)\bigg]}{E\langle f\rangle - v\langle f^2\rangle}
\end{eqnarray}
    
Writing the equation (\ref{NR}) as
\begin{eqnarray*}
\big(R(r)\big)E^2 - 2\big(vf(r)R(r)\big)E + \big(R^{\prime\prime}(r) + \frac{1}{r}R^\prime(r) + v^2 f^2(r)R(r) - m^2R(r) = 0,
\end{eqnarray*}
we obtain a quadratic equation of $E$. Thus
\begin{eqnarray}\label{EC}
E = vf(r)\pm\sqrt{m^2 - \frac{R^{\prime\prime}(r)}{R(r)} - \frac{R^\prime(r)}{rR(r)}}.
\end{eqnarray}
\medskip
\begin{lemma}
There exists an interval $J\subset [0,\infty)$ such that $ - \dfrac{R^{\prime\prime}(r)}{R(r)} - \dfrac{R^\prime(r)}{rR(r)}>0$.
\end{lemma}
\begin{proof}

\begin{enumerate}

\item {\bf $R$ is a node-free state:}
    
$R^{\prime\prime}(r) = 0$ $\iff$ $m^2 - \big(E-V(r)\big)^2 - \dfrac{\R^\prime(r)}{rR(r)} = 0$
$\iff$ $V(r) = E\pm\sqrt{m^2 - \dfrac{R^\prime}{rR(r)}}$.
    
If $R$ is decreasing near $0$, then $-\dfrac{R^\prime(r)}{rR(r)}>0$ near $0$.
    
If $R$ is increasing near $0$, then it must attain a maximum value at some $r_0\in[0,\infty)$ and end up decreasing since $\displaystyle\lim_{r\to\infty}R(r) = 0$. Thus, $-\dfrac{R^\prime}{rR(r)}>0$ on $(r_0,\infty)$.
    
Hence, in both cases $R$ must be decreasing on a subset $(r_0,\infty)$ of $[0,\infty)$, and $\sqrt{m^2-\dfrac{R^\prime(r)}{rR(r)}}>m$ on this subset interval.
    
Therefore, $V$ cannot take the value $ E+\sqrt{m^2-\dfrac{R^\prime(r)}{rR(r)}}$ since $V$ is non-positive and 
\begin{eqnarray}\label{RV}
R^{\prime\prime}(r) = 0\iff V(r) = E-\sqrt{m^2-\frac{R^\prime(r)}{rR(r)}}.
\end{eqnarray}
    
Let $r_i$ be a root of equation (\ref{RV}). 
    
If $r_i\in(r_0,\infty)$, then $J = (r_0,r_i)$.
    
If $r_i\notin (r_0,\infty)$, then there must exist at least another inflection point $r_j\in (r_0,\infty)$ because $R$ vanishes at infinity, which also implies that $R>0, R^\prime<0$, and $R^{\prime\prime}<0$ on $(r_0,r_j)$. Therefore, $J=(r_0,r_j)$.
\item {\bf $R$ is an excited State:}
Suppose that $R$ has $n$ nodes $\alpha_1, \alpha_2,...\alpha_n$ and consider the interval $(\alpha_n,\infty)$.
    
Then 
\begin{eqnarray}\label{INF}
R^{\prime\prime}=0 \iff m^2-\big(E-V(r)\big)^2-\dfrac{R^\prime(r)}{rR(r)}=0.
\end{eqnarray}
If $R$ is increasing near $\alpha^+$, then it should attain a maximum value at some $r_0\in (\alpha_n,\infty)$, become decreasing, and change its concavity at $r_i\in(r_0,\infty)$, where $r_i$ is a root of equation (\ref{INF}), since $\lim_{x\to\infty}{R(r)}=0.$
Hence, $R>0, R^\prime<0,$ and $R^{\prime\prime}<0$ on $(r_0, r_i)$ and therefore $J=(r_0,r_i)$.
    
If $R$ is decreasing near $\alpha^+$, then by the same reasoning we conclude that $R<0, R^\prime>0$, and $R^{\prime\prime}>0$ on  $(r_0, r_i)$ and $J=(r_0,r_i)$.
\end{enumerate}
\end{proof}
\noindent Since we have proven the existence of an interval $J$ such that $ - \dfrac{R^{\prime\prime}(r)}{R(r)} - \dfrac{R^\prime(r)}{rR(r)}>0$, and since $|E|<m$, then the option $E=vf(r)-\sqrt{m^2 - \dfrac{R^{\prime\prime}(r)}{R(r)} - \dfrac{R^\prime(r)}{rR(r)}}$ in expression (\ref{EC}) is falsified.
    
Therefore
\begin{eqnarray}\label{ET}
E = vf(r) + \sqrt{m^2 - \frac{R^{\prime\prime}(r)}{R(r)} - \frac{R^\prime(r)}{rR(r)}}.
\end{eqnarray}
\medskip
\begin{thm}
\begin{eqnarray*}
f_1(r)\leq f_2(r)\Longrightarrow v_1\leq v_2,
\end{eqnarray*}
for all $r\in [0,\infty)$.
\end{thm}
\begin{proof}
Using the expression (\ref{ET}) in equation (\ref{va}) we get
\begin{eqnarray*}
v_a = \dfrac{v\int_0^\infty \bigg[(f_2(r) - f_1(r))\sqrt{m^2 - \dfrac{R^{\prime\prime}(r)}{R(r)} - \dfrac{R^\prime(r)}{rR(r)}}R^2(r)\bigg]dr}{\int_0^\infty\bigg[\sqrt{m^2 - \dfrac{R^{\prime\prime}(r)}{R(r)} - \dfrac{R^\prime(r)}{rR(r)}}R^2(r)\bigg]dr} > 0.
\end{eqnarray*}
 Hence, the proof is complete.
\end{proof}
\section{Square-Well spectral bounds for general bounded potential shapes }
In this section we exhibit a complete recipe for finding square-well potential bounds for any {\it bounded} potential shape $f$ in the class considered in the previous sections, and consequently, spectral bounds for the coupling $v$, provided the energy is fixed.
We have chosen the square-well potential because we know the analytical solution for the Klein--Gordon problem with this potential. Before showing this solution, we state the following lemma:
\begin{lemma}
Consider the $d$-dimensional Klein--Gordon equation ($d\geq 0$)
\begin{eqnarray}\label{kn} 
\varphi^{\prime\prime}(r) =\left[m^2 - \big(E - V(r)\big)^2 + \frac{Q}{r^2} \right]\varphi(r),
\end{eqnarray}
 where $V(r) = vf(r)$ and $f$ belongs to the class of potential shapes defined in the previous sections. We define $s>0$ and $E_1$ to be the new energy corresponding to the potential $V_1(r) = v(f(r)-s)$. Then $|E + vs| < m$ and $E_1 = E - vs$. 
\end{lemma}
\begin{proof}
For $r\to+\infty$, the Klein--Gordon equation becomes $\varphi^{\prime\prime}(r) =\left[m^2-\big(E+vs\big)^2\right]  \varphi(r)$; thus, $\varphi(r) = C_1e^{kr} + C_2e^{-kr}$ with $k = \sqrt{m^2 - (E + vs)^2}$. Since $\varphi$ vanishes at $\infty$, then $C_1 = 0$, and since $\varphi\in L^2(\mathbb{R})$, then $|E + vs| < m$. Moreover, we can write \ref{kn} as:    
$\varphi^{\prime\prime}(r) = \left[m^2 - \big(E - vs - V(r) + vs \big)^2 + \frac{Q}{r^2} \right]\varphi(r) = \left[m^2 - \big(E - vs - v(f(r) - vs)\big)^2 + \frac{Q}{r^2}\right]\varphi(r)    
 =\left[m^2 - \big((E - vs) - V_1(r)\big)^2 + \frac{Q}{r^2}\right]\varphi(r)$. Therefore, $E_1 = E - vs$.    
\end{proof}

\subsection{A compact recipe for general spectral bounds}
\noindent Consider an attractive potential $V(r) = vf(r)$, where $f$ is a bounded potential shape in the class defined in the previous sections. We want to find the best square-well spectral bounds for the graph $v = G(E)$.
We define the downward vertically-shifted square-well potential 
\[ g(r,t_1) = 
\begin{cases}
           f(0), &  r \leq t_1  \\  
           f(t_1), & {\rm elsewhere}
      \end{cases}
,\]
with $s>0$, and the square-well potential 
\[ g(r,t_2) = 
\begin{cases}
           f(t_2), &  r \leq t_2   \\ 
           0, & {\rm elsewhere}
      \end{cases}
.\]
Thus, $g(r,t_1)\leq f(r)\leq g(r,t_2)$ for all $r\geq 0$, and for each pair of contact points $\{t_1,t_2\}$.    
We observe that $f(r)$ has infinite families of lower and upper bounds $G_L^{(t_1)}(E)\leq G(E)\leq G_U^{(t_2)}(E)$, where $G_L^{(t_1)}(E)$ and $G_U^{(t_2)}(E)$ are the respective spectral functions $v_L(E)$ and $v_U(E)$.    
The final step is to optimize over the parameter $t$ in order to obtain the best square-well spectral bounds for $G(E)$, that is
\begin{eqnarray}\label{bound}
G_L(E) = \displaystyle\max_{t_1>0} G_L^{(t_1)}(E)\leq G(E) \leq G_U(E) = \displaystyle\min_{t_2>0}G_U^{(t_2)}(E).
\end{eqnarray}
These functions are extracted from the eigenvalue equations \ref{SolE} and \ref{SolO} for the one-dimensional case, and from \ref{sphbessel} and \ref{ricbessel} in the higher dimensional cases. For example, we consider a square-well potential with depth $A$ and semi-width $b$ in dimension $d = 1$. Define the new variables $e = Eb$, $u = Ab$, $\mu = mb$, and $t = b\big[(E + A)^2 - m^2\big]^{\frac{1}{2}}$. Then from equation \ref{SolE} the ground state solution  becomes:
\begin{eqnarray*}
e(t) = \pm\big[\mu^2 - \big(t\cdot\tan(t)\big)^2\big]^{\frac{1}{2}} \text{ and }\   u(t) = (t^2+\mu^2)^{\frac{1}{2}} - e(t).
\end{eqnarray*}
For definiteness, we now assume $\mu = 1$. We observe that $e = 0$ when $t = t_0 \approx 0.860334$. The graph depicting $u = G(e)$ is shown in Figure $2$:
\begin{figure}[htbp]
\centering
\includegraphics[scale=0.4]{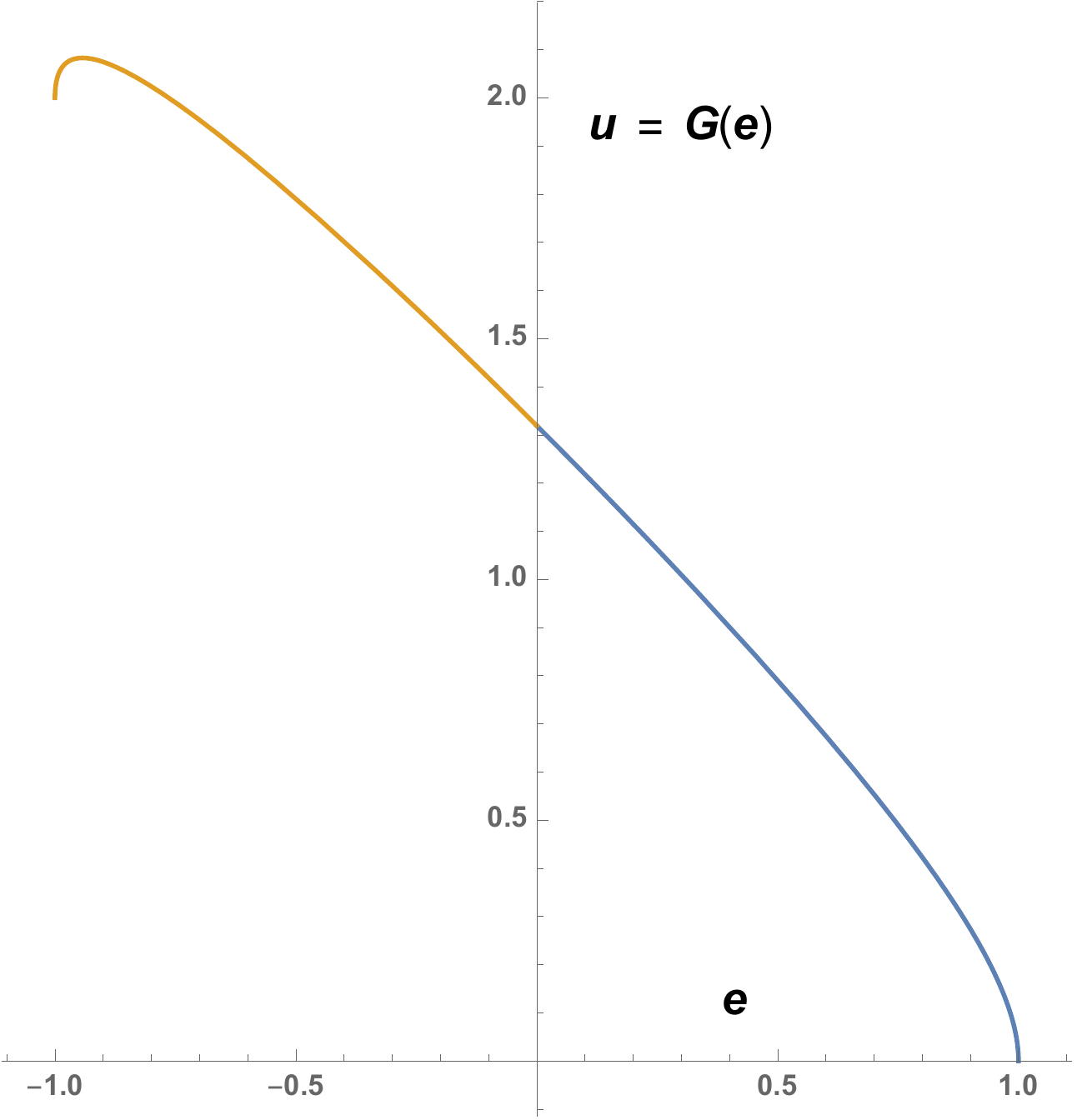}
\caption{u versus e}
\end{figure}
\subsection{The Woods-Saxon potential in $1$ - dimension}

\noindent We consider the Woods-Saxon potential $V(x) = vf(x)$, where $f(x) = -1\left(1 + e^{\frac{(|x| - 1)}{q}}\right)^{-1}$, and $q > 0$ is a range parameter. We are interested in finding an upper bound and a lower bound for the coupling constant $v$, for any given value of $|E| < m$ and for $q = 0.005$. Since the Klein--Gordon equation with the square-well potential had been solved analytically, we use a square-well potential as an upper bound for $f$, and another downward vertically-shifted square-well as a lower bound. We define the functions
\[ g_u(x,0.9675) = 
\begin{cases}
           -0.9984, &  |x| \leq 0.9675\\
           0, & {\rm elsewhere}
      \end{cases}
,\] 
and
\[ g_l(x,1.03) = 
\begin{cases}
           -1.001, &  |x| \leq 1.03\\
           -0.0025, & {\rm elsewhere}
      \end{cases}
.\]
Since $f_l(x) \leq f(x) \leq f_u(x)$ for all $x\in(-\infty,+\infty)$, then according to theorem III.$1$, we conclude that $G_L(E) =  v_ l\leq v \leq G_U(E) = v_u$, where $v_l$ and $v_u$ are the respective couplings for $f_l$ and $f_u$.    For example, if we fix $E = -0.512574196$, we get $v_u = 1.81478$ and $v_l = 1.79017$. Hence we conclude that $1.79017\leq v\leq 1.81478$. This result has been verified numerically, using our own shooting method realized in Maple, and with which we find $v = 1.80494$.
\begin{figure}[htbp]
\centering
\includegraphics[scale=0.4]{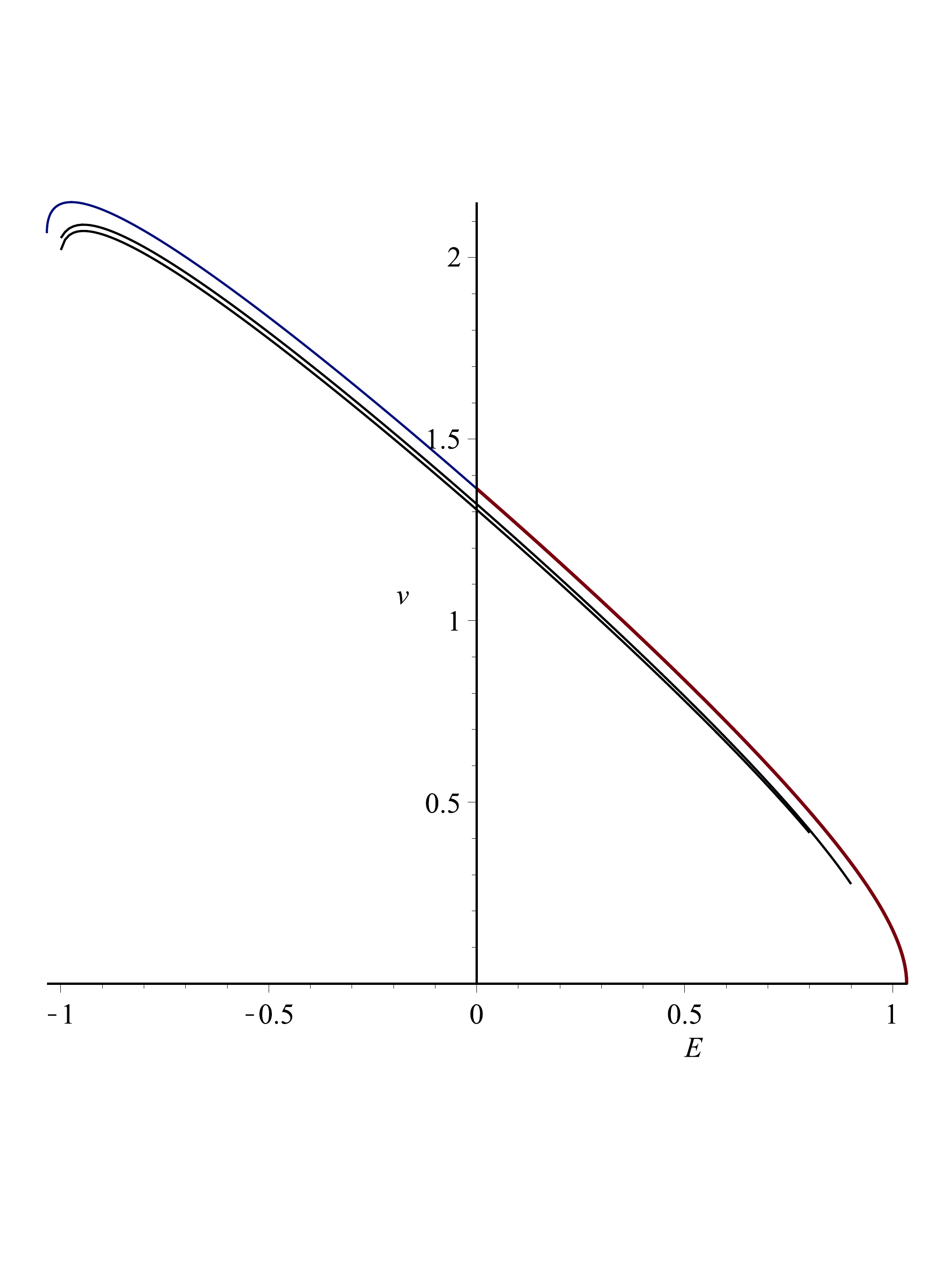}
\caption{Graphs for $v_l$, $v$, and $v_u$ versus $E$ for $-1 < E < 1$.}
\end{figure}
\begin{figure}[htbp]
\centering
\includegraphics[scale=0.35]{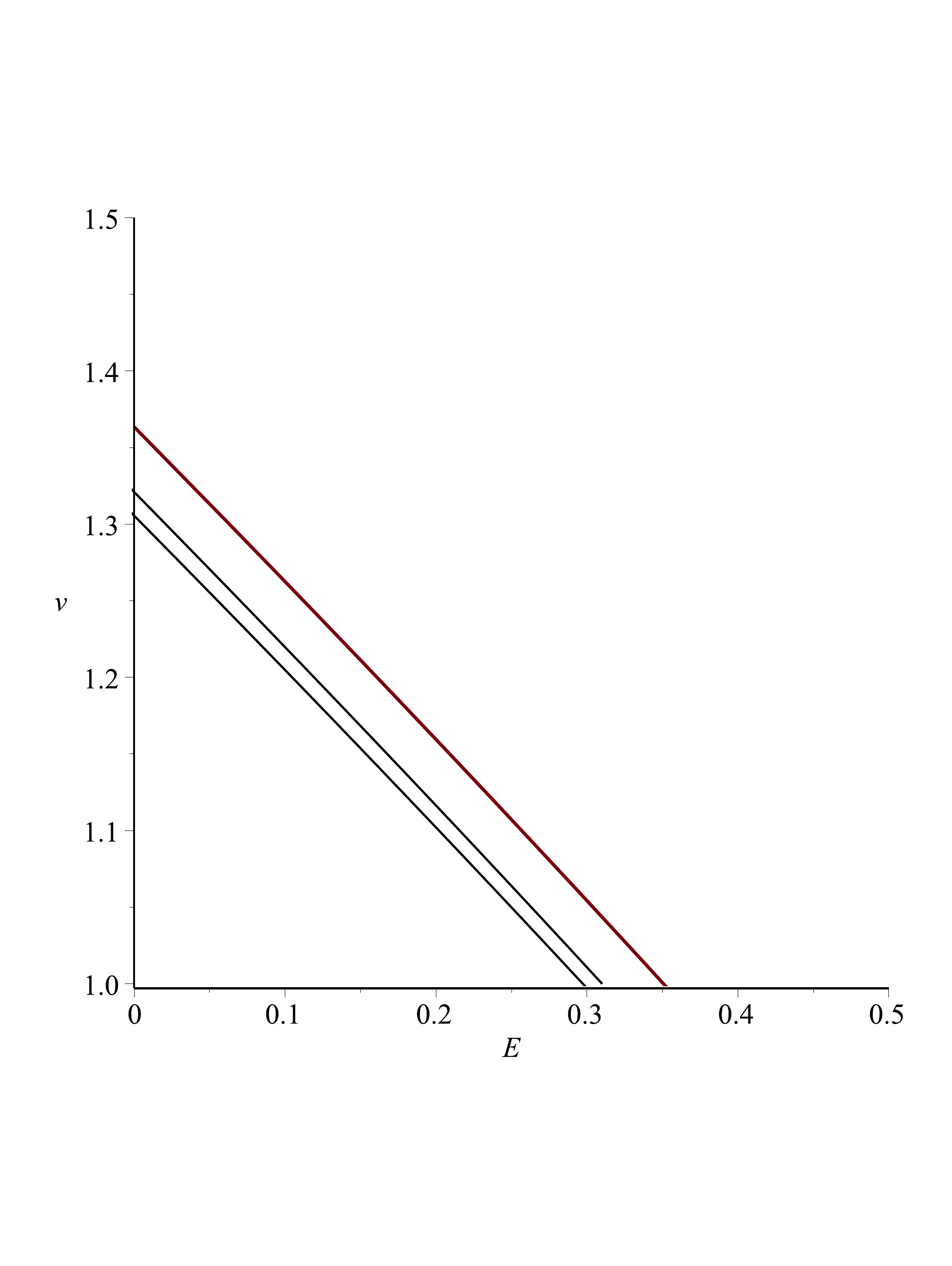}
\caption{Graphs for $v_l$, $v$, and $v_u$ versus $E$ for $0 < E < 0.5$.}
\end{figure}
\begin{figure}[htbp]
\centering
\includegraphics[scale=0.35]{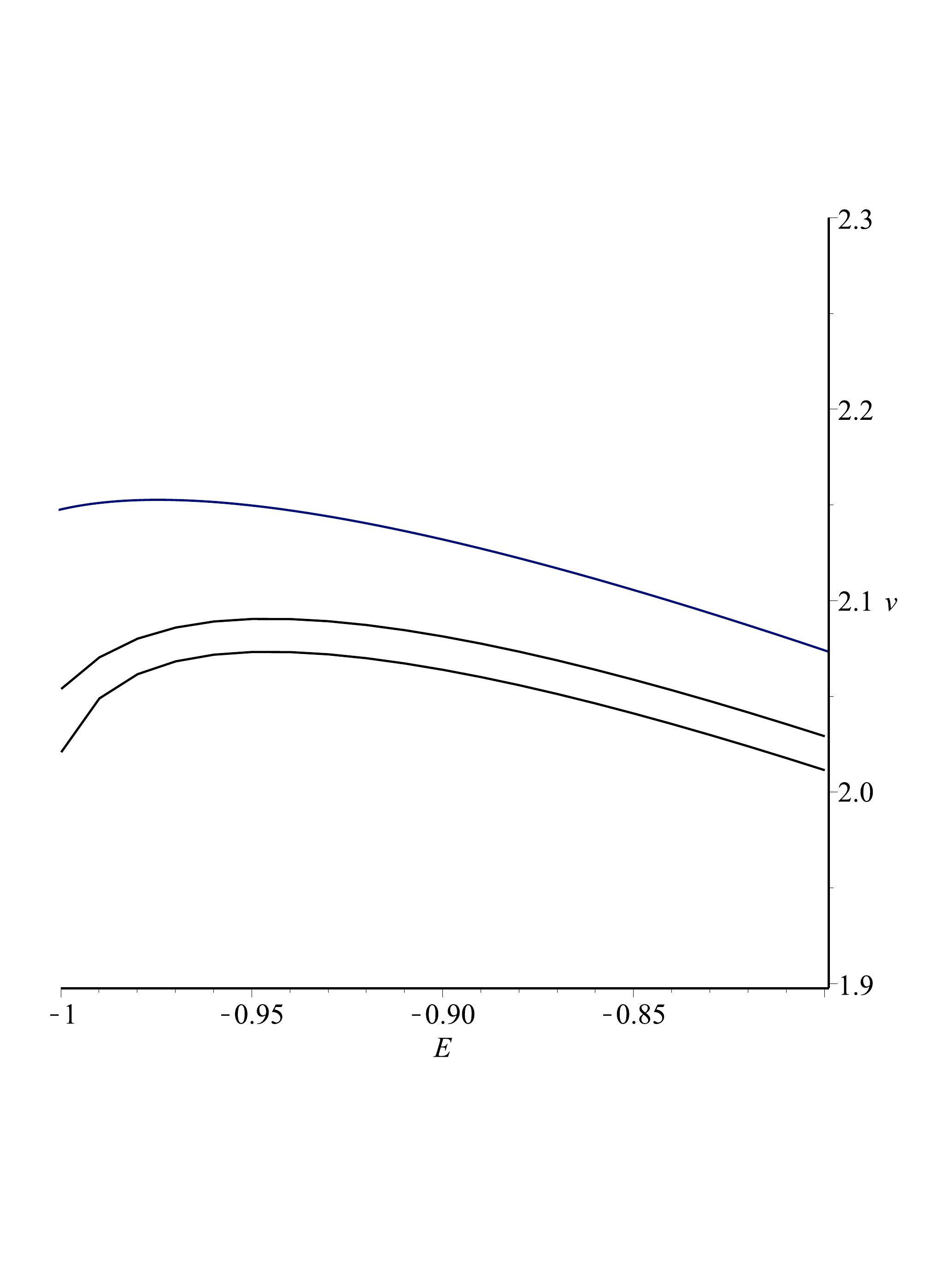}
\caption{Graphs for $v_l$, $v$, and $v_u$ versus $E$ for $-1 < E < -0.8$.}
\end{figure}
\newpage
\section{Conclusion}
The radial reduced eigenequations for a one-particle potential model might in suitable units be written,  for the non-relativistic and Klein-Gordon cases respectively, as:
\begin{itemize}
\item{(NR)}
\begin{eqnarray*}
\varphi^{\prime\prime}(r)=\bigg[ (2m)(v\,f(r)-E)+\frac{Q}{r^2}\bigg]  \varphi(r),
\end{eqnarray*}
\item{(KG)}
\begin{eqnarray*}
\varphi^{\prime\prime}(r)=\bigg[ m^2-\big(E-v\,f(r)\big)^2+\frac{Q}{r^2}\bigg]  \varphi(r), 
\end{eqnarray*}
\end{itemize}
where $Q=\frac{1}{4}(2l+d-1)(2l+d-3),$ the potential has shape $f(r) <0$ and coupling parameter $v > 0$.   We note that  a slightly different formulation of the Klein--Gordon  equation is required if $d = 2$ and $\ell = 0.$  By familiarity with well-known Schr\"odinger  examples, or by a variational analysis of them we expect, for suitable $v> v_0>0$, to find bound states with nonrelativistic energies $E(v)$
having monotonic behaviour $E'(v) < 0$ if the potential shape $f(r)$ is negative.  However, these assumptions are   not  correct for the corresponding Klein--Gordon eigenvalues.  This makes it difficult  to design  physically realistic potential models for relativistic problems.
\medskip

In this paper, we first represent the relation between the coupling $v$ and a discrete Klein-Gordon  eigenvalue $E$ by writing $v$ as a function $v=G(E)$ of $E$ for  $-m < E < m$. We show generally  that the spectral curve $v=G(E)$ is concave, and at most unimodal with a maximum close to $E = -m.$  For the purpose of comparing the  spectral implications of a change in
the potential shape, a bridging parameter $a \in [0,1]$ is introduced such that $f = f_1+a(f_2 - f_1)$.  By studying the dependence of $v$ on $a$ for each fixed value of $E$, we establish the comparison theorem $f_1\leq f_2\Longrightarrow G_{1}(v) \leq G_{2}(v)$. These results are valid for all negative and positive eigenenergies, and for both ground and excited states. They allow us to devise spectral approximations in much the same way as is possible for the corresponding Schr\"odinger  problem where the discrete spectrum can be defined variationally and the concomitant comparison theorems follow almost automatically by means of variational arguments.  As an illustration, we are able to use the exact solution of the square-well problem to construct upper and lower bounds for the discrete Klein--Gordon spectrum generated by any given member of the class of bounded negative potentials that we have considered in the present study.

 \begin{acknowledgments}
Partial financial support of his research under Grant No.~GP3438 from~the Natural
Sciences and Engineering Research Council of Canada is gratefully acknowledged.

 \end{acknowledgments}


\newpage

\begin{thebibliography}{99}
\bibitem{reedsimon}M. Reed and B. Simon, {\it Methods of Modern Mathematical Physics IV: Analysis of Operators}, (Academic Press, New York, 1978).
\bibitem{Fr} J. Franklin and L. Intemann, {\it Saddle-Point Variational Method for the Dirac Equation}, Phys. Rev. Lett. {\bf 54}, 2068 (1985).
\bibitem{Gold} S. P. Goldman, {\it Variational Representation of the Dirac-Coulomb Hamiltonian with no spurious roots},  Phys. Rev. A {\bf 31}, 3541 (1985).
\bibitem{Gr} I. P. Grant and H. M. Quiney, {\it Rayleigh-Ritz approximation of the Dirac operator in atomic and molecular physics}, Phys. Rev. A {\bf 62}, 022508 (2000).
\bibitem{ssw}  
L. I. Schiff, H. Snyder, and J. Weinberg, {\it On The Existence of Stationary States of the Mesotron Field} ,    
Phys. Rev. {\bf 57}, 315 (1940).
\bibitem{Gr1}W. Greiner, {\it Relativistic Quantum Mechanics: Wave Equations} third ed. (Springer, Berlin, 2000) p. 59.
\bibitem{bl}M. Bawin and J.P. Lavine, {\it The Exponential Potential and the Klein--Gordon Equation}, Phys. Rev. D {\bf 12}, 1192 (1975).
\bibitem{Gr2}W. Greiner, {\it Relativistic Quantum Mechanics: Wave Equations} third ed. (Springer, Berlin, 2000) p. 61.
\bibitem{vc}V. M. Villalba and C. Rojas, {\it Bound States of the Klein--Gordon Equation in the Presence of Short Range Potentials}, International Journal of Modern Physics A, Vol. {\bf 21}, pp. 313-325 (2006).
\bibitem{RH-99} R. L. Hall, {\it Spectral Comparison Theorem for the Dirac Equation}, Phys. Rev. {\bf 83}, 468 (1999).
\bibitem{RH-08} R. L. Hall, {\it Special Comparison Theorem for the Dirac Equation}, Phys. Rev. {\bf 101}, 090401 (2008).
\bibitem{RP-15} R. L. Hall and P. Zorin, {\it Refined Comparison Theorems for the Dirac Equation in $d$ Dimensions}, Ann. Phys. (Berlin) {\bf 527}, 408-422 (2015).
\bibitem{R:P-16} R. L. Hall and P. Zorin, {\it Refined Comparison Theorems for the Dirac Equation with Spin and Pseudo--Spin Symmetry in $d$ Dimensions}, Eur. Phys. J. Plus. {\bf 131}, (2016) 102.
\bibitem{RA-08} R. L. Hall and M. D. Aliyu, {\it Comparison Theorems for the Klein--Gordon Equation in $d$ Dimensions}, Phys. Rev. A {\bf 78}, 052115 (2008).
\bibitem{RH-10} R. L. Hall, {\it Relativistic Comparison Theorems}, Phys. Rev. A {\bf 81}, 052101 (2010).
\bibitem{RP-16} R. L. Hall and P. Zorin, {\it Sharp Comparison Theorems for the Klein--Gordon Equation in $d$ Dimensions}, Int. J. Mod. Phys. E {\bf 25} (2016) 1650039.
\bibitem{rtd}R. Friedberg, T. D. Lee, and W. Q. Zhao, {\it Convergent iterative solutions for a Sombrero-shaped potential in any space dimension and arbitrary angular momentum}, Ann. Phys. {\bf 321}, 1981 (2006).
\bibitem{MMN}M. M. Nieto, {\it Hydrogen atom and relativistic pi-mesic atom in N-space dimensions}, Am. J. Phys. {\bf 47}, 1067 (1979).
\bibitem{grnr} W. Greiner, {\it Relativistic Quantum Mechanics: Wave Equations} third ed. (Springer, Berlin, 2000) p: 57.
\bibitem{abw} M. Abramowitz and I. A.  Stegun, Editors, {\it Handbook of Mathematical Functions}, (U. S. Government Printing Office, Washington DC, 1964). 

\end{thebibliography}
\end{document}